\documentclass[11pt,oneside]{article}
\usepackage[english]{babel}
\usepackage[colorlinks]{hyperref}
\usepackage{amsthm,amsmath,amsfonts,prettyref,paralist,fullpage}

\newtheorem{thm}{Theorem}
\newtheorem{lem}[thm]{Lemma}

\DeclareMathOperator\Inf{Inf}
\DeclareMathOperator*\Var{Var}
\DeclareMathOperator*\Ent{Ent}
\DeclareMathOperator\BG{BG}
\DeclareMathOperator\NTW{\texttt{NTW}}
\DeclareMathOperator*\E{\mathbb{E}}
\DeclareMathOperator\D{\mathbb{D}}
\DeclareMathOperator\Sch{Sch}
\newcommand\R{\mathbb{R}}
\newcommand\Z{\mathbb{Z}}

\begin{document}
\title{Hypercontractivity and its Applications}
\author{Punyashloka Biswal}
\maketitle

\begin{abstract}
  Hypercontractive inequalities are a useful tool in dealing with extremal questions in the geometry of high-dimensional
  discrete and continuous spaces. In this survey we trace a few connections between different manifestations of
  hypercontractivity, and also present some relatively recent applications of these techniques in computer science.
\end{abstract}

\section{Preliminaries and notation}

\paragraph*{Fourier analysis on the hypercube.}

We define the inner product $\langle f,g \rangle = \E_{x}f(x)g(x)$ on functions $f,g \colon \{-1,1\}^{n} \to \R$, where
the expectation is taken over the uniform (counting) measure on $\{-1,1\}^n$. The multilinear polynomials
$\chi_{S}(x)=\prod_{i\in S}x_{i}$ (where $S$ ranges over subsets of $[n]$) form an orthogonal basis under this inner
product; they are called the Fourier basis. Thus, for any function $f \colon \{-1,1\}^{n}\to\R$, we have $f =
\sum_{S\subseteq[n]}\hat{f}(S)\chi_{S}(x)$, where the Fourier coefficients $\hat{f}(S)=\langle f,\chi_{S}\rangle$ obey
Plancherel's relation $\sum\hat{f}(S)^{2}=1$. It is easy to verify that $\E_{x}f(x)=\hat{f}(0)$ and
$\Var_{x}f(x)=\sum_{S\neq\emptyset}\hat{f}(S)^{2}$.

\paragraph*{Norms.}

For $1\leq p<\infty$, define the $\ell_{p}$ norm $\|f\|_{p}=(\E_{x}|f(x)|^{p})^{1/p}$.  These norms are monotone in $p$:
for every function $f$, $p\geq q$ implies $\|f\|_{p}\geq\|f\|_{q}$. For a linear operator $M$ carrying functions $f
\colon \{-1,1\}^{n}\to\R$ to functions $Mf=g \colon \{-1,1\}^{n}\to\R$, we define the $p$-to-$q$ operator norm
$\|M\|_{p\to q}=\sup_{f}\|Mf\|_{q}/\|f\|_{p}$.  $M$ is said to be a contraction from $\ell_{p}$ to $\ell_{q}$ when
$\|M\|_{p\to q}\leq1$. Because of the monotonicity of norms, a contraction from $\ell_{p}$ to $\ell_{p}$ is
automatically a contraction from $\ell_{p}$ to $\ell_{q}$ for any $q<p$. When $q>p$ and $\|M\|_{p\to q}\leq1$, then $M$
is said to be hypercontractive.

\paragraph*{Convolution operators.}

Letting $xy$ represent the coordinatewise product of $x, y \in \{-1,1\}^n$, we define the convolution
$(f*g)(x)=\E_{y}f(x)g(xy)$ of two functions $f,g \colon \{-1,1\}^{n}\to\R$, and note that it is a linear operator
$f\mapsto f*g$ for every fixed $g$.  Convolution is commutative and associative, and the Fourier coefficients of a
convolution satisfy the useful property $\widehat{f*g}=\hat{f}\hat{g}$.  We shall be particularly interested in the
convolution properties of the following functions
\begin{itemize}
\item The Dirac delta $\delta \colon \{-1,1\}^{n}\to\R$, given by $\delta(1,\dotsc,1)=1$ and $\delta(x)=0$ otherwise. It
  is the identity for convolution and has $\hat{\delta}(S)=1$ for all $S\subseteq[n]$.
\item The edge functions $h_{i} \colon \{-1,1\}^{n}\to\R$ given by\[ h_{i}(x)=
  \begin{cases}
  \phantom{-}1/2 & x=(1,\dotsc,1)\\
  -1/2 & x_{i}=-1,x_{[n]\setminus\{i\}}=(1,\dotsc,1)\\ 
  \phantom{-}0 & \text{otherwise.}
  \end{cases}
  \] $\hat{h}_{i}(S)$ is $1$ or $0$ according as $S$ contains or does not contain $i$, respectively. For any function $f
  \colon \{-1,1\}^{n}\to\R$, $(f*h_{i})(x)=(f(x)-f(y))/2$, where $y$ is obtained from $x$ by flipping just the $i$th
  bit. Convolution with $h_{i}$ acts as an orthogonal projection (as we can easily see in the Fourier domain), so for
  any functions $f,g \colon \{-1,1\}^{n}\to\R$, we have $\langle f*h_{i},g\rangle=\langle f,h_{i}*g\rangle=\langle
  f*h_{i},g*h_{i}\rangle$
\item The Bonami-Gross-Beckner noise functions $\BG_{\rho} \colon \{-1,1\}^{n}\to\R$ for $0\leq\rho\leq1$, where
  $\widehat{\BG}_{\rho}(S)=\rho^{|S|}$ and we define $0^{0}=1$. These operators form a semigroup, because
  $\BG_{\sigma}*\BG_{\rho}=\BG_{\sigma\rho}$ and $\BG_{1}=\delta$.  Note that
  $\BG_{\rho}(x)=\sum_{S}\rho^{|S|}\chi_{S}(x)=\prod_{i}(1+\rho x_{i})$.  We define the noise operator $T_{\rho}$ acting
  on functions on the discrete cube by $T_{\rho}f=\BG_{\rho}*f$. In combinatorial terms, $(T_{\rho}f)(x)$ is the
  expected value of $f(y)$, where $y$ is obtained from $x$ by independently flipping each bit of $x$ with probability
  $1-\rho$.
\end{itemize}
\begin{lem}
$\frac{d}{d\rho}\BG_{\rho}=\frac{1}{\rho}\BG_{\rho}*\sum
  h_{i}$
\end{lem}
\begin{proof}
  \label{lem:noise-deriv}This is easy in the Fourier basis:
  \[
  \widehat{\BG}_{\rho}'
  = (\rho^{|S|})'
  = |S|\rho^{|S|-1}
  = \sum_{i\in[n]}\hat{h}_{i}\frac{\widehat{\BG}_{\rho}}{\rho}.\qedhere
  \]
\end{proof}

\section{The Bonami-Gross-Beckner Inequality}

\subsection{Poincar\'e and Log-Sobolev inequalities}

The Poincar\'e and logarithmic Sobolev inequalities both relate a function's global non-constantness to how fast it
changes ``locally''. The amount of local change is quantified by the \emph{energy} $\D(f,f)$, where the Dirichlet form
$\D$ is defined as
\[
\D(f,g)=\tfrac12 \E_{xy\in E}(f(x)-f(y))(g(x)-g(y))
\]
($E$ is the set of pairs $x,y$ that differ in a single coordinate).  In terms of the edge functions $h_{i}$, observe
that $\D(f,g)=\frac2n \sum_{i}\langle f*h_{i},g*h_{i}\rangle$.

In the case of the Poincar\'e inequality, we measure the distance of $f$ to a constant by its variance $\Var(f)=\E(f-\E
f)^{2}=\E f^{2}-(\E f)^{2}$.  Then the Poincar\'e constant (of the discrete cube) is the supremal $\lambda$ such that the
inequality
\[
\D(f,f)\geq\lambda\Var(f)
\]
holds for all $f \colon \{-1,1\}^{n}\to\R$. This quantity is also the smallest nonzero eigenvalue of the Laplacian of
the discrete cube, viewed as a graph (i.e., its spectral expansion).

Another way of measuring the non-constantness of a function is to consider its entropy $\Ent(f)=\E[f\log \frac{f}{\E
  f}]$ (where we assume $f\geq0$ and use the convention that $0\log0=0$). Note that $\Ent(cf)=c\Ent(f)$ for any
$c\geq0$, so the entropy is homogenous of degree $1$ in its argument. Because we are comparing the entropy with the
energy (which is homogenous of degree $2$) we use the entropy of the \emph{square} of the function to define the
Log-Sobolev constant: the largest $\alpha$ such that the inequality\[ \D(f,f)\geq\alpha\Ent(f^{2})\] holds for all $f
\colon \{-1,1\}^{n}\to\R$.
 For the discrete cube $\{-1,1\}^{n}$, we have $\lambda=2/n$ and $\alpha=1/n$, as we shall
see below. It is interesting to ask how these quantities are related when we consider other probability spaces equipped
with a suitable Dirichlet form (for example, $d$-regular graphs with $\D(f,g)=\E_{xy\in E}(f(x)-f(y))(g(x)-g(y))$, where
the expectation is taken over all edges). Set $f=1+\epsilon g$ for a sufficiently small $\epsilon$ and observe that
$\Var(f)=\epsilon^{2}\Var(g)$ and $\D(f,f)=\epsilon^{2}\D(g,g)$, whereas
\begin{align*}
  \Ent(f^{2}) &= \E\left[(1+\epsilon g)^{2}(2\log(1+\epsilon g)
    -\log\E[(1+\epsilon g)^{2}])\right]\\
  &=2\epsilon^{2}\Var(g)+O(\epsilon^{3})
\end{align*}
This shows that $\alpha\leq\lambda/2$, which is tight in the case of the cube.  However, for constant-degree expander
families (in particular, for random $d$-regular graphs with high probability) we have \cite[Example 4.2]{DSC}
$\lambda=\Omega(1)$ but $\alpha=O(\log\log n/\log n)\ll\lambda$.

\subsection{Hypercontractivity and the log-Sobolev inequality}

When $\rho\in[0,1]$, the noise operator $T_{\rho}$ is easily seen to contract $\ell_{2}$: for any $f \colon \{-1,1\}^{n}
\to \R$, we have $\|T_{\rho}f\|_{2}^{2} = \sum_{S}\rho^{|S|}\hat{f}(S)^{2} \leq \sum_{S} \hat{f}(S)^{2} =
\|f\|_{2}^{2}$.  Now consider its behavior from $\ell_{2}$ to $\ell_{q}$ for some $q>2$. When $\rho=1$, we have
$T_{1}f=f$; in particular, for $g(x) = (1+x_{1})/2$, $\|g\|_{q} = 1/2^{1/q}>1/2^{1/2} = \|g\|_{2}$. On the other hand,
$T_{0}f=\E f$, so $\|T_{0}f\|_{q}=|\E f|\leq\|\E f^{2}\|^{1/2}$.  By the intermediate value theorem, there must be some
$\rho\in(0,1)$ such that $\|T_{0}\|_{2\to q}=1$.  A theorem of Gross \cite{Gross} connects this critical $\rho$ with the
Log-Sobolev constant $\alpha$ of the underlying space:
\begin{thm}\label{thm:gross}
  $\|T_{\rho} f\|_{p \to q} \leq 1$ if and only if $\rho^{-2 \alpha n} \geq\frac{q-1}{p-1}$.
\end{thm}
Stated differently, $\|T_{1-\epsilon}f\|_{q}\leq\|f\|_{2}$ when $q\leq(1-\epsilon)^{-2}+1\approx 2 + 2\epsilon$.  Thus
to prove hypercontractive inequalities on the discrete cube, it suffices to bound the log-Sobolev constant.  We shall
prove this claim for $p=2$, which turns out to imply the general version.

\begin{proof}[Proof of~\prettyref{thm:gross}]
  We shall prove that $\|T_{\rho}f\|_{q}\leq\|f\|_{2}$ for $q=1+\rho^{-2\alpha n}$; the remainder of the theorem can be
  shown using similar techniques.  As we observed before, this inequality is tight when $\rho=1$, so it suffices to show
  that $\frac{d}{d\rho}\|T_{\rho}f\|_{q}\geq0$ for $0\leq\rho\leq1$. For notational convenience, let
  $G=\|T_{\rho}f\|_{q}^{q}$.  Then
  \[
  \|T_{\rho}f\|_{q}'=(G^{1/q})'=q^{-2}G^{(1/q)-1}\left(qG'-q'G\log
  G\right).
  \]
  Now we use the fact that $G=\E(T_{\rho}f)^{q}$ to get
  \[
  G'=q\E\left[(T_{\rho}f)^{q-1}(T_{\rho}f)'\right]+q'\E\left[(T_{\rho}f)^{q}
    \log(T_{\rho}f)\right].
  \]
  Applying \prettyref{lem:noise-energy} and simplifying, we get
  \[
  qG'-q'G\log G=q'\Ent\left((T_{\rho}f)^{q}\right)
  +\frac{nq^{2}}{2\rho}\D\left((T_{\rho}f)^{q-1},T_{\rho}f\right).
  \]
  We use \prettyref{lem:dirichlet-convexity} to handle the second term, and plug in $q=1+\rho^{-2\alpha n}$ to get
  \[
  qG'-q'G\log G=n \rho^{-2\alpha n - 1} \bigl[ \D\bigl( (T_{\rho}f)^{q/2},(T_{\rho}f)^{q/2}\bigr) -
  \Ent\left((T_{\rho}f)^{q}\right)\bigr],
  \]
  whose positivity we are guaranteed by the log-Sobolev inequality applied to $(T_{\rho}f)^{(q-1)/2}$.
\end{proof}
\begin{lem}
  \label{lem:noise-energy}For any $f,g \colon \{-1,1\}^{n}\to\R$, $\langle g, \frac{d}{d\rho}(T_{\rho}f)\rangle
  =\frac{n}{2\rho}\D(g,T_{\rho}f)$.
\end{lem}
\begin{proof}
  Recalling \prettyref{lem:noise-deriv} and the projection property of the $h_{i}$s, we have
  \[
  \langle g,(T_{\rho}f)'\rangle
  = \langle g,\BG_{\rho}'*f\rangle
  = \biggl< g, \frac{1}{\rho} \BG_{\rho} * f * \sum_{i} h_{i} \biggr>
  = \frac{1}{\rho}\sum_{i}\langle g*h_{i},\BG_{\rho}*f\rangle
  = \frac{n}{2\rho}\D(g,T_{\rho}f).
  \qedhere
  \]
\end{proof}
\begin{lem}
  \label{lem:dirichlet-convexity}For any $f \colon \{-1,1\}^{n}\to\R$ and $q\geq2$,
  $\D(f,f^{q-1})\geq\frac{4(q-1)}{q^{2}}\D\left(f^{q/2},f^{q/2}\right)$.
\end{lem}
\begin{proof}
  It suffices to show that $(a^{q-1}-b^{q-1})(a-b)>\frac{4(q-1)}{q^{2}}(a^{q/2}-b^{q/2})^{2}$ for all $a>b\geq0$ and
  $q\geq2$. But observe that\begin{align*}
    \left(\int_{a}^{b}t^{q/2-1}dt\right)^{2} & =\frac{4}{q^{2}}(a^{q/2}-b^{q/2})^{2}\\
    \int_{a}^{b}t^{q-2}dt\ \int_{a}^{b}dt & =\frac{1}{q-1}(a^{q-1}-b^{q-1})(a-b)\end{align*} and the inequality between
  the integrals follows from convexity.
\end{proof}

\subsection{Two-point inequality}

We begin by showing that the log-Sobolev inequality holds for the uniform distribution on the two-point space $\{-1,1\}$
with $\alpha=2$. Without loss of generality, consider $f(x)=1+sx$. Then
\[
\Ent(f^{2})=\tfrac12 (1+s)^{2}\log(1+s)^2+\tfrac12 (1-s)^{2}\log(1-s)^2 - (1+s^{2})\log(1+s^{2})
\]
and $\D(f,f)=2s^{2}$. We shall show that $\phi(s)=\D(f,f)-\alpha\Ent(f^{2})$ is non-negative for $-1\leq s\leq1$. By
symmetry it suffices to consider $s\geq0$. But $\phi(0)=0$ and
\[
\phi'(s) = 4s + 2s\log(1+s^{2}) + 2(1-s)\log(1-s) - 2(1+s)\log(1+s),
\]
which is non-negative because $\phi'(0)=0$ and
\[
\phi''(s)=\frac{4s^{2}}{s^{2}+1}+2\log\frac{1+s^{2}}{1-s^{2}}\geq0.
\]

\subsection{Tensoring property}

\begin{thm}\label{thm:tensoring}
  Let $\alpha$ be the log-Sobolev constant of $\{-1,1\}^n$. Then the log-Sobolev constant of $\{-1,1\}^{2n}$ is
  $\alpha/2$.
\end{thm}
When $n$ is a power of $2$, we can conclude inductively that $\alpha=1/n$; a proof along similar lines works for
arbitrary $n$ as well.
\begin{proof}[Proof of~\prettyref{thm:tensoring}]
  For any $f \colon \{-1,1\}^{n}\times\{-1,1\}^{n}\to\R$, set $g(x)=\|f(x,\cdot)\|_{2}$.  Then by the conditional
  entropy formula,
  \[
    \Ent(f^{2})
    \leq \Ent(g^{2})+\E_{x}\Ent_{y}(f(x,y)^{2})
    \leq \frac{\D(g, g) + \E_{x}\D_{y}(f(x,y), f(x, y))}{\alpha}
  \]
  and by convexity,
  \begin{align*}
    \D(g, g) &=
    \tfrac12 \E_{x \sim x'}(g(x)-g(x'))^{2} \leq
    \tfrac12 \E_{x \sim x'}\E_{y}\left[(f(x,y)-f(x',y))^{2}\right]
    = \E_{y}\D_x(f(x,y), f(x, y))
  \end{align*}
  where the notation $x \sim x'$ ranges over edges of $\{-1,1\}^{n}$. Taken together, these give
  \[
  \Ent(f^{2})
  \leq \frac{\E_{x}\D_y(f(x,y), f(x,y)) + \E_{y}\D_x(f(x,y), f(x, y))}{\alpha}
  \leq \frac{2\D(f)}{\alpha}
  = \frac{\D(f)}{\alpha/2}.
  \]
  as claimed.
\end{proof}

\subsection{Non-product groups}

Recall that we defined the Dirichlet form
\[
\D(f, g) = \tfrac{1}{2} \E_{u\sim v} (f(u) - f(v))(g(u) - g(v))
\]
for functions $f, g \colon \{-1,1\}^n \to \R$, but it makes sense for any regular graph if we sample $u, v$ uniformly
from the edges. Thus, given any family of regular graphs, we can ask if they satisfy a log-Sobolev inequality of the
form $\D(f, f) \geq \alpha \Ent(f)$ for all suitable $f$.

It turns out that the relationship between logarithmic Sobolev inequalities and hypercontractive noise operator
subgroups, as stated by Gross \cite{Gross}, holds for a wide class of spaces, not just the hypercube
$\{-1,1\}^n$. Diaconis and Saloff-Coste \cite{DSC} explored an intermediate between these two extremes of
specialization to give improved mixing time results for Markov chains on various graphs.

One of the first discrete applications of hypercontractivity was a celebrated theorem of Kahn, Kalai and Linial
\cite{KKL} relating the maximum influence of a function on the hypercube to its variance. In \prettyref{thm:schreier},
we discuss some recent work \cite{OWim} of O'Donnell and Wimmer generalizing the KKL theorem to apply to the wider class
of Schreier graphs associated with group actions (defined below).

An action of a group $G$ on a set $X$ is a homomorphism from $G$ to the group of bijections on $X$, and we write $x^{g}$
for the image of $x$ under the bijection for $g$. If $S$ is a set of generators for $G$, then the Schreier graph
$\Sch(G,S,X)$ has vertex set $X$ and edges $(x,x^{g})$ for all $x\in X$ and $g\in S$. It is known that every connected
regular graph of even degree can be obtained in this way \cite{GrossSch}. The definition of the Dirichlet form $\D$
generalizes without change, but to be able to derive a log-Sobolev inequality for this space, we must define the noise
operator $T_{\rho}$ in an appropriate fashion to satisfy the claim of \prettyref{lem:noise-deriv}: $\langle
g,\frac{d}{d\rho}(T_{\rho}f)\rangle \propto \frac{1}{\rho}\D(g,T_{\rho}f)$.

\section{Boolean-Valued Functions}

\subsection{Influences}

Write $x_{-i}$ for the collection of random variables $\{x_{1},\dotsc,x_{n}\}\setminus\{x_{i}\}$.  The influence of the
$i$th coordinate on a function $f \colon \{-1,1\}^{n}\to\R$ is given by 
\[
\Inf_{i}(f)=\E_{x_{-i}}\Var_{x_{i}}f(x)=
\E_{x_{-i}}\left[\E_{x_{i}}f(x)^{2}-(\E_{x_{i}}f(x))^{2}\right].
\]
When $f$ is Boolean-valued, this quantity is just the probability that changing $x_{i}$ changes $f(x)$. Writing $f$ in
the Fourier basis, we have $\E_{x_{-i}}\E_{x_{i}}f(x)^{2}= \E_{x}f(x)^{2}= \sum_{S}\hat{f}(S)^{2}$ and
$\E_{x_{-i}}(\E_{x_{i}}f(x))^{2}= \sum_{S\not\ni i}\hat{f}(S)^{2}$, so that $\Inf_{i}(f)=\sum_{S\ni i}\hat{f}(S)^{2} =
\E(f * h_i)^2$. In addition, we define the total influence $\Inf(f)=\sum_{i}\Inf_{i}(f)=\sum_{S}|S|\hat{f}(S)^{2}$.

\subsection{Structural results}

Boolean functions are natural combinatorial objects, but they were first studied from an analytical viewpoint in work on
voting and social choice. In this setting, a function $f \colon \{-1,1\}^{n}\to\{-1,1\}$ is viewed as a way to combine
the preferences of $n$ voters to yield the result of the election. This explains the notions of dictator or junta
functions, which depend on only one or a few of their coordinates, respectively. In this context it is also natural to
consider functions where no coordinate (``voter'') has a very large influence. Kahn, Kalai, and Linial \cite{KKL} first
introduced the Fourier analysis of Boolean functions as a technique in computer science. Their theorem establishes that
if a function is far from a constant (i.e.,~has variance at least a constant), then it must have a variable of influence
$\Omega(\frac{\log n}{n})$.  We state a strengthening of their original inequality due to Talagrand \cite{Talagrand}:
\begin{thm}[\cite{KKL,Talagrand}]
  For any $f \colon \{-1,1\}^{n}\to\{-1,1\}$,
  \[
  \sum_{i}\frac{\Inf_{i}(f)}{\log(1/\Inf_{i}(f))}\geq
  \Omega(1)\cdot\Var(f).
  \]
\end{thm}
We can compare this to the Poincar\'e inequality on the cube, which can be stated as
\[
\sum_{i}\Inf_{i}(f)\geq\Omega(1)\cdot\Var(f).
\]
(In particular, there exists a variable of influence $\Omega\bigl(\tfrac{1}{n}\bigr)\Var(f)$.) The KKL theorem is a stronger
result of the same form: it is a comparison between a local and a global measure of variation. The proofs of KKL and
Talagrand used the hypercontractivity of the cube, but we present here a more recent proof due to Rossignol that uses
the log-Sobolev inequality instead. For simplicity we'll just show the weaker statement that the maximum influence is
$\Omega\bigl(\tfrac{\log n}{n}\bigr)\Var(f)$.
\begin{proof}
  Write $f-\E f=f_{1}+\dotsb+f_{n}$, where $f_{j}=\sum_{S:\max S=j}\hat{f}(S)\chi_{S}$.  For each $f_{j}$, the
  log-Sobolev inequality states that $\D(f_{j},f_{j}) \geq \alpha\Ent(f_{j}^{2}) = \frac1n \Ent(f_{j}^{2})$.  By writing
  $\D(f_{j},f_{j})$ in terms of the Fourier coefficients $\hat{f}(S)$, we can check that
  $\D(f,f)=\sum_{j=0}^{n}\D(f_{j},f_{j})$, so that we can sum all these inequalities to obtain
  \[
  n \D(f,f)\geq\sum_{j}\Ent(f_{j}^{2})=
  \underbrace{\sum_{j}\E\left[f_{j}^{2}\log(f_{j}^{2})\right]}_{A}
  +\underbrace{\sum_{j}\E f_{j}^{2}\log\frac{1}{\E f_{j}^{2}}}_{B}.
  \]
  In order to bound $B$, we begin by noting that
  \[
  \E f_{j}^{2}=\sum_{S:\max S=j}\hat{f}(S)^{2}\leq\sum_{S\ni
    j}\hat{f}(S)^{2}=\E(f*h_{j})^{2}
  \]
  where the $h_{j}$s are the edge functions we defined earlier. Letting $M(f) = \max_{j}\E(f*h_{j})^{2} = \max_j
  \Inf_j(f)$, we have
  \[
  B = \sum_{j} \E f_{j}^{2} \log \frac{1}{\E f_{j}^{2}} \geq \sum_j \E f_{j}^{2} \log \frac{1}{M(f)} =
  \Var(f)\log\frac{1}{M(f)}
  \]
  where we have used the orthogonality of the $f_{j}$s and the fact that $\Var(f)=\sum_{S\neq\emptyset} \hat f(S)^2$.

  To bound $A$, we split it up further:
  \[
  A=\underbrace{\sum_{j}\E\left[f_{j}^{2}
      \log(f_{j}^{2})\cdot1_{f_{j}^{2}\leq t}\right]}_{A_{1}}+
  \underbrace{\sum_{j}\E\left[f_{j}^{2}
      \log(f_{j}^{2})\cdot1_{f_{j}^{2}>t}\right]}_{A_{2}}.
  \]
  For $0 \leq t\leq1/e^{2}$, we have that $\sqrt{t} \log \sqrt{t}$ is a nonpositive decreasing function and therefore,
  \[
  A_{1} = 2 \sum_{j}\E\left[|f_{j}|\log|f_{j}| \cdot |f_{j}| 1_{f_{j}^{2} \leq t}\right] \geq 2 \sqrt{t}\log \sqrt{t}
  \sum_{j} \E| f_{j} \cdot 1_{f_{j}^2 \leq t} | \geq \sqrt{t} \log t \sum_{j} \E|f_{j}|.
  \]
  By comparing Fourier coefficients, it is easy to verify that $f_{j}=\E_{x_{j+1},\dotsc,x_{n}}(f*h_{j})$. Therefore, by
  convexity, $\E|f_{j}|\leq\E|f*h_{j}|.$

  Until now, the proof has made no use of the fact that $f$ takes on only Boolean values. Now we argue that because
  $f(x)\in\{-1,1\}$, we must have $(f*h_{j})(x)\in\{-1,0,1\}$, so that $\E|f*h_{j}|=\E(f*h_{j})^{2}$.  Plugging this
  into our bound for $A_{1}$ yields
  \[
  A_{1} \geq \sqrt{t} \log t \sum_{j} \E(f*h_{j})^{2} = \frac{n}{2} \sqrt{t} \log t \cdot \D(f,f).
  \]
  For $A_{2}$, note that $\log(\cdot)$ is increasing, so
  \[
  A_{2} \geq \log t\sum_j\E f_{_{j}}^{2}=\log t\Var f.
  \]
  Summing all these bounds gives us
  \[
  n \D(f, f) \geq \log \frac1{M(f)} \Var(f) + \frac n2 \sqrt t \log t \cdot \D(f, f) + \log t \cdot \Var(f).
  \]
  By the Poincar\'e inequality, $\D(f, f) \geq \frac 2n \Var(f)$, so we can set $t = \bigl( \frac{2\Var(f)}{ne\D(f,f)}
  \bigr)^2 \leq 1/e^2$. With this substitution, the above inequality becomes
  \[
  \frac{2}{e\sqrt t} \geq \log \frac{t^{1+1/e}}{M(f)}.
  \]
  Suppose $t \leq (\frac{4}{e\log n})^2$. Then
  \[
  \D(f, f) \geq \frac{2 \Var(f)}{en} \cdot \frac{e\log n}{4} = \Omega\Bigl(\frac{\log n}{n}\Bigr),
  \]
  and we know that $M(f) \geq 2 \D(f, f)$.  On the other hand, if $t > (\frac{4}{e \log n})^2$, then
  \[
  M(f) > t^{1+1/e} \exp\Bigl(\frac{-2}{e \sqrt t}\Bigr)
  = \Bigl(\frac{4}{e \log n}\Bigr)^{2+2/e} \exp \Bigl(\frac{-\log n}{2} \Bigr) \gg \frac{\log n}{n}. \qedhere
  \]
\end{proof}
We are now in a position to state the recent result of O'Donnell and Wimmer \cite{OWim} generalizing the KKL theorem to
Schreier graphs satisfying a certain technical property.
\begin{thm}[\cite{OWim}]\label{thm:schreier}
  Let $G$ be a group acting on a set $X$, $U\subseteq X$ be a union of conjugacy classes that generates $G$, and
  $\alpha$ be the log-Sobolev constant of $\Sch(G,X,U)$. Then for any $f \colon X\to\{-1,1\}$,
  \[
  \frac{\sum_{U}\Inf_{u}(f)}{\log(1/\max_{U}\Inf_{u}(f))}\geq
  \Omega(\alpha\Var(f)).
  \]
  In particular, there is some $u\in U$ such that
  $\Inf_{u}(f)\geq\Omega(\alpha \log\frac{1}{\alpha})\Var(f)$.
\end{thm}
For an Abelian group such as $\Z_{2}^{n}$ (the cube), every group element is in a conjugacy class by itself, so the
extra condition on $U$ is vacuous. Using $\alpha=\Omega(\frac{1}{n})$ for the cube, we recover the original KKL
theorem. O'Donnell et al.~apply the generalized result to the non-Abelian group $S_{n}$ of permutations on $[n]$,
generated by transpositions and acting on the family $\binom{[n]}{k}$ of $k$-subsets of $[n]$. By viewing these families
as sets of $n$-bit strings, they recover a ``rigidity'' version of the Kruskal-Katona theorem that states (roughly) that
if a subset of a layer of a cube has a small expansion to the layer above it, then it must be correlated to some
dictator function.

\paragraph*{Coding theoretic interpretation.} In the \emph{long code}, an integer $i \in [n]$ is encoded as the dictator
function $(x_1, \dotsc, x_n) \mapsto x_i$. By using many more bits ($2^n$ rather than $\log n$) of redundant storage, we
hope to be able to recover from corruptions in the data. The theorem tells us that as long as the corrupted version of
an encoding is far from a constant function, it can be decoded to a coordinate whose influence is $\Omega(\log n)$ times
the average influence. Since every coordinate's influence is nonnegative, only $O(\log n)$ coordinates can have
influence this large. Thus, we have a {}``small'' set of candidate long codes to which we might decode the word. To
complete this picture, we'd like to understand how far the word can be from functions that depend only on these
coordinates; the following theorem of Friedgut, which we state without proof, furnishes this information.
\begin{thm}[\cite{Friedgut}]
  For every $f \colon \{-1,1\}^{n}\to\{-1,1\}$ and $0<\epsilon<1$, there is a function $g \colon
  \{-1,1\}^{n}\to\{-1,1\}$ depending on at most $\exp\bigl(\frac{2 + o(1)}{\epsilon n} \Inf(f)\bigr)$ variables such
  that $\E|f-g|\leq\epsilon$.
\end{thm}

\section{Gaussian isoperimetry and an algorithmic application}

Hypercontractive inequalities were first investigated in the context of Gaussian probability spaces, for their
applications to quantum field theory. The following simple proof reduces the continuous Gaussian hypercontractive
inequality to its discrete counterpart on the cube.

\subsection{From the central limit theorem to Gaussian hypercontractivity}
\begin{thm}[\cite{Gross}]
  Let $x\in\R$ be normally distributed, i.e.,
  \[
  \Pr[x\in A]= \frac{1}{\sqrt{2\pi}}\int_{A}
  \exp\left(-\frac{x^{2}}{2}\right)\,dx.
  \]
  Then for a smooth function $f \colon \R\to\R$, the random variable $F=f(x)$
  satisfies
  \[
  \D(F,F)\geq\alpha\Ent(F^{2})
  \]
  with $\alpha=1$ and
  \[
  \D(F,G)=\frac12 \left<{\frac{dF}{dx},\frac{dG}{dx}}\right> .
  \]
\end{thm}
\begin{proof}
  We shall approximate the Gaussian distribution by a weighted sum of Bernoulli variables. Let $y\in\{-1,1\}^{k}$ be
  uniformly distributed, and set $g(y)=\frac{y_{1}+\dotsb+y_{k}}{\sqrt{k}}$. By the log-Sobolev inequality applied to
  $f\circ g(y)$, we have $\D\left(f\circ g(y),f\circ g(y)\right)\geq\Ent(f\circ g(y)^{2})$.  By the central limit
  theorem, the right side converges to $\Ent(f(x)^{2})=\Ent(F^{2})$ as $k\to\infty$, so it remains to show that the left
  side converges to $\D(F,F)$ as well. Let $y|_{y_{i}=\theta}$ be the value obtained by replacing the $i$th coordinate
  of $y$ with the value $\theta$, and observe that $g(y|_{y_{i}=1})-g(y|_{y_{i}=-1})=2/\sqrt{k}$. Then, using the
  smoothness of $f$, we have
  \[
  \left|\left(h_{i}*(f\circ g)\right)(y)\right|=
  \frac{1}{2}\left|f\circ g(y|_{y_{i}=1})-f\circ
  g(y|_{y_{i}=-1})\right|=\frac{1}{\sqrt{k}}\left|f'\circ
  g(y)\right|+o\left(\frac{1}{\sqrt{k}}\right),
  \]
  so that
  \[
  \D\left(f\circ g(y),f\circ g(y)\right)
  = \frac12 \E_{y}\left[\sum_{i}\left(h_{i}*(f\circ g)\right)(y)^{2}\right]
  = \frac12 \E_{y}\left[f'\circ g(y)^{2}+o(1)\right].
  \]
  The second term vanishes as $k\to\infty$, and the first term converges to $\D(F,F)$ by the Central Limit Theorem.
\end{proof}

The tensoring property of log-Sobolev inequalities lets us extend this result to Gaussian distributions over $\R^d$. We
are also interested in the corresponding noise operator $S_\rho$, known as the Ornstein-Uhlenbeck operator, which is
given by
\[
S_\rho f(x) = \E_{z \sim \mathcal{N}(0,1)^d} f(\rho x + (1-\rho^2)^{1/2} z).
\]
\prettyref{thm:gross} has an analog in this setting, which lets us conclude that every function $f \colon \R^d \to \R$
satisfies $\| S_\rho f \|_q \leq \| f \|_p$ where $q > p \geq 1$ and $\rho^{-2} \geq (p-1)/(q-1)$.
\subsection{Reverse hypercontractivity and isoperimetry}

In 1982, Borell showed a \emph{reversed} inequality of a similar form when $q < p < 1$:

\begin{thm}[Reverse hypercontractivity, \cite{Borell}]
  Fix $q < p \leq 1$ and $\rho \geq 0$ such that $\rho^{-4} \geq (p-1)/(q-1)$. Then for any positive-valued function $f
  \colon \R^d \to \R^+$, we have $\| S_\rho f \|_q \geq \| f \|_p$.
\end{thm}

Note that the expressions $\| \cdot \|_p$ are not norms when $p < 1$; in particular, they are not convex. However, this
theorem can be proved by means similar to our proof for the Gaussian log-Sobolev inequality: we start with a base result
for the 2-point space, proceed by tensoring to the hypercube, and use the central limit theorem to cover Gaussian
space.

As an application of Borell's result, consider the following strong isoperimetry theorem for Gaussian space
(due to Sherman).

\begin{thm}[Gaussian isoperimetry, \cite{Sherman}]
  Let $u, u' \in \R^d$ be independent Gaussian random variables. Then for any set $A \subseteq \R^d$ and any $\tau > 0$,
  we have
\[
\Pr_u \left[ \Pr_{u'}[\rho u + (\sqrt{1-\rho^2}) u' \in A] \leq \tau \right] \leq \frac{\tau^{1-\rho}}{\mu(A)}
\]
\end{thm}

\begin{proof}
  When $\mu(A) \leq \tau^{1-\delta}$, there is nothing to prove. Otherwise, let $f$ be the indicator function of $A$ and
  observe that $\Pr_{u'} \bigl[\rho u + (1-\rho^2)^{1/2} u' \in A \bigr] = S_\rho f(u)$. Therefore, for $q = 1-1/\rho < 0$, we
  have
\begin{align*}
\Pr_u\left[\Pr_{u'} [u' \in A] \leq \tau\right]
&= \Pr_u [S_\rho f(u) \leq \tau] \\
&= \Pr_u [S_\rho f(u)^q \geq \tau^q] \\
&\leq \frac{\E_u (S_\rho f(u))^q}{\tau^q}
\end{align*}
by an application of Markov's inequality. But $\E_u (S_\rho f(u))^q$ is just $\| S_\rho f \|_q^q$, and we know by
Borell's theorem that $\| S_\rho f \|_q \geq \| f \|_p$ for $p = 1-\rho$. Thus
\[
\Pr_u\left[\Pr_{u'} [u' \in A] \leq \tau\right]
\leq \frac{\| f \|_p^q}{\tau^q}
= \frac{\mu(A)^{q/p}}{\tau^q}
= \left(\frac{\tau^{1-\rho}}{\mu(A)}\right)^{1/\rho}
\leq \frac{\tau^{1-\rho}}{\mu(A)}
\]
where we have used the facts that $q < 0$ and $\rho \leq 1$.
\end{proof}

\subsection{Fast graph partitioning and the constructive Big Core Theorem}

\paragraph*{Problem and SDP rounding algorithm.} In the $c$-balanced separator problem, we are given a graph $G = (V,
E)$ on $n$ vertices and asked to find the smallest set of edges such that their removal disconnects the graph into
pieces of size at most $cn$. The problem is NP-hard, and the best known approximation ratio\footnote{For technical
  reasons, it is actually a pseudo-approximation: the algorithm's output for $c$ is compared to the optimal value for
  $c' \neq c$.} is $\Theta(\sqrt{\log n})$.

The first algorithm to achieve this bound was based on a semidefinite program that assigns a unit vector to each vertex
and minimizes the total embedded squared length of the edges subject to the constraint that the vertices are spread out
and that the squared distances between the points form a metric:
\begin{align*}
  \text{minimize} & \textstyle\sum_{i \sim j} \| x_i - x_j \|_2^2 \\
  \text{subject to}
  & \| x_i \|_2^2 = 1 & \forall i \in V \\
  & \textstyle\sum_{i,j} \| x_i - x_j \|^2 \geq c(1-c)n \\
  & \| x_i - x_j \|_2^2 + \| x_j - x_k \|_2^2 \geq \| x_i - x_k \|^2 & \forall i, j, k \in V
\end{align*}
To round this SDP, Arora, Rao and Vazirani \cite{ARV} pick a random direction $u$ and project all the points
along $u$. They then define sets $A$ and $B$ consisting of points $x$ whose projections are sufficiently large, i.e., $A
= \{ x \mid \langle x, u \rangle < -K \}$ and similarly $B = \{ x \mid \langle x, u \rangle > K \}$, where $K$ is chosen
to make $A$ and $B$ have size $\Theta(n)$ with high probability. Next, they discard points $a \in A, b \in B$ such that
$\| a - b \|$ is much smaller than expected for a pair whose projections are $\geq 2K$ apart. Finally, if the resulting
pruned sets $A' \subset A$ and $B' \subset B$ are large enough, they show that greedily growing $A$ yields a good cut.

\paragraph*{Matchings and cores.} The key step in making this argument work is to ensure that not too many pairs $(a,b)$
are removed in the pruning step. To bound the probability of this bad event, we consider the possibility that for a
large fraction $\delta = \Omega(1)$ of directions $u$, there exists a matching of points $M_u$ such that each pair $(a,
b) \in M_u$ is short (i.e., $\| a - b \| \leq \ell = O(1/\sqrt{\log n})$) but stretched along $u$ (i.e., $| \langle a -
b, u \rangle | \geq \sigma = \Omega(1)$). Such a set of points is called a \emph{$(\sigma, \delta, \ell)$-core}. The big
core theorem (first proved with optimal parameters by Lee \cite{Lee}) asserts that this situation can't arise: for a
fixed $\sigma, \delta$, and $\ell$, we must have $n \gg \exp(\sigma^6/\ell^4\log^2(1/\delta))$, which is a contradiction
for our chosen values of $\sigma, \delta, \ell$.

In order to prove the big core theorem, Lee concatenates pairs that share a point and belong in matchings for nearby
directions. The existence of a long chain of such concatenations is what leads to a contradiction: if we consider the
endpoints $a, b$ of a chain of length $p$, the projection $| \langle a - b, u \rangle |$ grows linearly in $p$ whereas
the distance $\| a - b \|$ grows only as $\sqrt p$ (recall that the SDP constrained the \emph{squared} distances to form
a metric).

\paragraph*{Boosting.} The matching chaining argument we have just presented in its simple form doesn't
work, for the following reason. At each chaining step, the fraction of nearby directions available for our use reduces
by roughly $1 - \delta$ (by a union bound) so that we are rapidly left with no direction to move in. To remedy this
situation, we need to boost the fraction of usable directions at each step, say from $\delta/2$ to $1 - \delta/2$, so
that we can carry on chaining in spite of a $1 - \delta$ loss. Lee's proof uses the standard isoperimetric inequality
for the sphere to show that this boosting can be performed with no change in $\ell$ and a very small penalty in
$\sigma$. In other words, we take advantage of the fact that a very small dilation of a set of constant measure (i.e.,
the set of available directions) has measure close to $1$.

\paragraph*{Faster algorithms.} Lee's big core theorem is non-constructive in the sense that it only shows the
\emph{existence} of such a long chain of matched pairs in order to give a contradiction. While this form suffices to
bound the approximation ratio of the ARV rounding scheme, other variants of their technique require a way to
\emph{efficiently sample} long chains, not just show their existence. Sherman constructs a distribution over directions
that does not depend on the point set at all, yet is guaranteed to always have a non-trivial probability of producing
long chains of stretched pairs. More precisely,

\begin{thm}[Constructive big core \cite{Sherman}]
  For any $1 \leq R \leq \Theta(\sqrt{\log n})$, there is $P \geq \Theta(R^2/\log n)$ and an efficiently sampleable
  distribution $\mu$ over the set of sequences of $\leq P$ direction vectors (each in $\R^d$), such that: for any
  $(\sigma, \delta, \ell)$-core $M$, if the string of directions is sampled from $\mu$, the expected number of chains
  whose endpoints are $\geq P\ell$ apart is at least $\exp(-O(P^2)n)$.
\end{thm}

We sketch some of the ideas of the proof here. Sherman constructs two sequences of Gaussian directions $u_1, \dotsc,
u_P$ and $w_1, \dotsc, w_P$. Each $w_i$ is an independent Gaussian vector, whereas each $u_i$ for $i > 1$ is a Gaussian
vector $\rho$-correlated with $u_{i-1}$. Finally, the distribution $\mu$ is given by randomly shuffling together the
$u_i$ and $w_i$, picking a uniformly random $R$ between $1$ and $P$, and returning the first $R$ elements of the
shuffled sequence. The correlated directions $u_i$ correspond to the steps in which Lee's proof chained pairs from
similar directions, whereas the independent $w_i$ correspond to the region-growing steps necessary for boosting. By
randomly interleaving these two types of moves, Sherman's sampling algorithm can be oblivious to the actual point set it
is acting on.

\section{Complexity theoretic applications}

\subsection{Dictatorship testing with perfect completeness}

\paragraph*{Definitions.} A function $f \colon \{-1,1\}^n \to \R$ is said to be $(\epsilon, \delta)$-quasirandom if
$\hat f(S) \leq \epsilon$ whenever $|S| \leq 1/\delta$.  In order to show that a given problem is hard to approximate,
we often need to design a test that
\begin{compactitem}
\item performs $q$ \emph{queries} on a black-box function $f$,
\item accepts every dictator function with probability $\geq c$ (the \emph{completeness} probability), and
\item accepts every $(\epsilon, \delta)$-quasirandom function with probability $\leq s$ (the \emph{soundness}
  probability).
\end{compactitem}
A test is said to be \emph{adaptive} if each query is allowed to depend on the result of the queries so far.

While dictatorship tests for the $c < 1$ setting have been known for over a decade (first from the work of H\aa stad and
more recently via the Unique Games Conjecture of Khot), there were no nontrivial bounds for $c = 1$ until some recent
results of O'Donnell and Wu. Their analysis, which we show below, relies heavily on the hypercontractive inequality.

\begin{thm}[\cite{OWu}] \label{thm:perfcomp}
  For every $n > 0$, there is a $3$-query non-adaptive test that accepts every dictator function $(x_1, \dotsc, x_n)
  \mapsto x_i$ with probability $c = 1$ but accepts any $(\delta, \delta/\log(1/\delta))$-quasirandom odd function $f
  \colon \{-1, 1\}^n \to [-1,1]$ with probability $\leq s = 5/8 + O(\sqrt\delta)$.
\end{thm}

The proof uses the following strengthening of the hypercontractive inequality for restricted parameter values.
\begin{lem}
  If $0 \leq \rho \leq 1$, $q \geq 1$, and $0 \leq \lambda \leq 1$ satisfy $\rho^\lambda \leq 1/\sqrt{q - 1}$, then for
  all $f \colon \{-1,1\}^n \to \R$, $\| T_\rho f \|_q \leq \| T_\rho f \|_2^{1-\lambda} \| f \|_2^\lambda$.
\end{lem}

\begin{proof}
  \begin{align*}
    \| T_\rho f \|_q^2
    &= \| T_{\rho^\lambda} T_{\rho^{1-\lambda}} f \|_q^2 \\
    &\leq \| T_{\rho^{1-\lambda}} f \|_2^2 \\
    &= \sum_S | \rho \hat f(S) |^{2(1-\lambda)} | \hat f(S)
    |^{2\lambda} \\
    &= \| T_\rho f \|_2^{2(1-\lambda)} \| f \|_2^{2\lambda} \qquad \qedhere
  \end{align*}
\end{proof}

\begin{proof}[Proof of \prettyref{thm:perfcomp}]
  Define the ``not-two'' predicate $\NTW \colon \{-1,1\}^3 \to \{-1, 1\}$ as follows: $\NTW(a, b, c) = 1$ if exactly two
  of $a, b, c$ equal $-1$, and $\NTW(a, b, c) = -1$ otherwise. Explicitly,
  \[
  \begin{array}{rrrrrrrrr}
    a             & -1 & -1 & -1 & -1 & 1 & 1 & 1 & 1 \\
    b             & -1 & -1 & 1 & 1 & -1 & -1 & 1 & 1 \\
    c             & -1 & 1 & -1 & 1 & -1 & 1 & -1 & 1 \\ \hline
    \NTW(a, b, c) & -1 & 1 & 1 & -1 & 1 & -1 & -1 & -1
  \end{array}
  \]
  Let $\delta \in [0, 1]$ be a parameter to be fixed later. For $i = 1, \dotsc, n$, we pick bits $x_i, y_i, z_i \in
  \{-1,1\}$ as follows:
  \begin{compactitem}
  \item with probability $1-\delta$: we choose $x_i, y_i$ uniformly and independently, then set $z_i = -x_i y_i$;
  \item with probability $\delta$: we choose $x_i$ uniformly, then set $y_i = z_i =
    x_i$.
  \end{compactitem}
  Note that for $i \neq j$, $(x_i, y_i, z_i)$ is independent of $(x_j, y_j, z_j)$.  We accept if $\NTW(f(x), f(y), f(z))
  = -1$. It is immediate from the construction of $x_i, y_i, z_i$ that $\NTW(x_i, y_i, z_i) = -1$ for $i = 1, \dotsc,
  n$. Therefore, if $f$ is a dictator function, it follows that $\NTW(f(x), f(y), f(z))$ must also equal $-1$.

  \paragraph*{Soundness.} It remains to analyze the test when $f$ is pseudorandom. We begin by writing $\NTW$ in the
  Fourier basis: $\NTW = -\frac14 \chi_\emptyset - \frac14(\chi_{\{1\}} + \chi_{\{2\}} + \chi_{\{3\}}) -
  \frac14(\chi_{\{1,2\}} + \chi_{\{2,3\}} + \chi_{\{1,3\}}) + \frac34\chi_{\{1,2,3\}}$. Therefore, by symmetry,
  \[
  \E_{x,y,z} \NTW(f(x), f(y), f(z))
  = -\tfrac14 -\tfrac34 \E_x f(x) -\tfrac34 \E_{x,y} f(x)f(y) + \tfrac34
  \E_{x,y,z} f(x) f(y) f(z).
  \]
  We shall systematically rewrite the right-hand side in terms of the Fourier coefficients of $f$. By our assumption
  that $f$ is odd, we have $\hat f(S) = 0$ whenever $S$ has even cardinality. Therefore $\E f(x) = \hat f(\emptyset) =
  0$.  Also,
  \[
  \E_{x,y} f(x) f(y)
  = \sum_{S,T} \hat f(S) \hat f(T) \E_{x,y} \chi_S(x) \chi_T(y).
  \]
  Consider a summand where $S \neq T$, and without loss of generality fix $i \in S \setminus T$. It is easy to see that
  the contributions due to $x_i = \pm1$ cancel each other. Thus, the only terms that remain are of the form $S = T$,
  i.e.,
  \[
  \E_{x,y} f(x) f(y)
  = \sum_S \hat f(S)^2 \E_{x,y} \chi_S(x) \chi_S(y)
  = \sum_S \hat f(S)^2 \left(\E_{x_i,y_i} x_i y_i\right)^{|S|}
  = \sum_S \hat f(S)^2 \delta^{|S|},
  \]
  where we have used the fact that $\E (x_i y_i) = (1 - \delta) \cdot 0 + \delta \cdot 1 = \delta$. But $\hat f(S)$ is
  nonzero only for $|S|$ odd, and $\sum_S \hat f(S)^2 = 1$, so we can upper-bound the above sum by $\delta$.

  \paragraph*{Bounding the cubic term.} We proceed similarly:
  \begin{equation} \label{eq:sum3}
    \E_{x,y,z} f(x) f(y) f(z) = \sum_{S,T,U} \hat f(S) \hat f(T) \hat f(U)
    \E_{x,y,z} \chi_S(x) \chi_T(y) \chi_U(z).
  \end{equation}
  Each of the expectations can be written as a product over coordinates $i \in [n]$ using the fact that individual
  coordinates of $x,y,z$ are chosen independently.  When $i$ belongs to exactly one of $S, T, U$ (say $S$), then it
  contributes a factor $\E x_i = 0$, making the product zero. Similarly, when $i$ belongs to two of the sets (say $S,
  T$), then the contribution is $\E x_i y_i = \delta$ by our earlier calculation. Finally, when $i$ belongs to all three
  of the sets, we have $\E x_i y_i z_i = (1-\delta)\cdot(-1) + \delta\cdot(0) = -(1-\delta)$. In light of this
  calculation, any triple $S, T, U$ that makes a nonzero contribution to the sum~\eqref{eq:sum3} must be of the form
  \begin{align*}
    S &= A \cup B \cup C &
    T &= A \cup C \cup D &
    U &= A \cup D \cup B
  \end{align*}
  for suitable sets $A, B, C, D \subseteq [n]$ where $A$ is disjoint from $B, C, D$. Also $|S|, |T|, |U|$ must be odd,
  from which we can show that $|A|$ must be odd. In terms of these new sets we can rewrite
  \[
  \E_{x,y,z} f(x) f(y) f(z) = -\hspace{-1em}\sum_{\substack{B, C, D \text{ disj.
        from } A\\|A| \text{ odd}}} \hspace{-1em}
  \hat f(A \cup B \cup C) \hat f(A \cup C \cup D) \hat f(A \cup D \cup B)
  (1-\delta)^{|A|} \delta^{|B| + |C| + |D|}.
  \]
  For a fixed $A$, define the function $g_A \colon \{-1,1\}^{[n] \setminus A} \to \R$ by $\hat g_A(X) = \hat f(A \cap
  X)$. Then we have
  \begin{align*}
    &\E_{x,y,z} f(x) f(y) f(z) \\
    &= -\sum_{|A|\text{ odd}} (1-\delta)^{|A|} \sum_{\substack{B,C,D\\\text{disj.
          from }A}} \hat g_A(B \cup C) \sqrt\delta^{|B\cup C|} \cdot \hat g_A(C \cup D)
    \sqrt\delta^{|C\cup D|} \cdot \hat g_A(D \cup B) \sqrt\delta^{|D\cup B|}\\
    &= -\sum_{|A|\text{ odd}} (1-\delta)^{|A|} \sum_{\substack{B,C,D\\\text{disj.
          from }A}} \widehat{T_{\sqrt \delta} g_A}(B \cup C) \cdot \widehat{T_{\sqrt
        \delta} g_A}(C \cup D) \cdot \widehat{T_{\sqrt \delta} g_A}(D \cup B) \\ &=
    -\sum_{|A|\text{ odd}} (1-\delta)^{|A|} \|T_{\sqrt \delta} g_A\|_3^3.
  \end{align*}
  Write $g_A(u) = \E_x g_A(u) + \tilde g_A(u) = \hat f(A) + \tilde g_A(u)$. Then, using the inequality $|a+b|^3 \leq
  4(|a|^3 + |b|^3)$, we have
  \[
  \| T_{\sqrt \delta} g_A \|_3^3 = \| \hat f(A) + T_{\sqrt \delta} \tilde g_A \|_3^3 \leq
  4 | \hat f(A) |^3 + 4 \| T_{\sqrt \delta} \tilde g_A \|_3^3
  \]
  and therefore,
  \begin{align*}
    \sum (1-\delta)^{|A|} \| T_{\sqrt \delta} g_A \|^3 \leq 4 \sum
    (1-\delta)^{|A|} | \hat f(A) |^3 + 4 \sum (1-\delta)^{|A|} \| T_{\sqrt \delta} \tilde
    g_A \|_3^3.
  \end{align*}
  To bound the first term, note that $\sum (1-\delta)^{|A|} |\hat f(A)|^3 \leq \sum \hat f(A)^2 \cdot
  \max\{(1-\delta)^{|A|} |\hat f(A)|)\}$. The sum of the squared Fourier coefficients is just $1$ (by Parseval's
  identity) and we can use the $(\delta, \frac{\delta}{\log(1/\delta)})$-pseudorandomness property to bound the quantity
  in the maximum: when $|A| < \frac1\delta \log \frac1\delta$, then $|\hat f(A)| \leq \sqrt\delta$ and when $|A| \geq
  \frac1\delta \log \frac1\delta$ then $(1-\delta)^{|A|} \leq \delta$. Thus the entire first summand is
  $O(\sqrt\delta)$.

  \paragraph*{Hypercontractivity.} It remains to bound $\sum (1-\delta)^{|A|} \| T_{\sqrt\delta}\tilde g_A \|_3^3$. Fix $\lambda =
  \frac{\log 2}{\log(1/\delta)}$ and apply the modified hypercontractive inequality:
  \begin{align*}
    \sum (1-\delta)^{|A|} \| T_{\sqrt\delta} \tilde g_A \|_3^3
    &\leq \sum (1-\delta)^{|A|} \| T_{\sqrt\delta} \tilde g_A
    \|_2^{3-3\lambda} \| \tilde g_A \|_2^{3\lambda}
  \end{align*}
  Now, $\| \tilde g_A \|_2^{3\lambda} \leq 1$ and $\| T_{\sqrt\delta} \tilde g_A \|_2^{3-3\lambda} = O(\sqrt \delta)
  \sum_{\emptyset \neq B \subseteq \overline A} \delta^{|B|} \hat f(A \cup B)^2$. The contribution of the corresponding
  term to the sum we were trying to bound is $O(\sqrt \delta) \cdot \hat f(A \cup B)^2 \cdot (1-\delta)^{|A|}
  \delta^{|B|}$. For each choice of $A \cup B$, the $(1-\delta)^{|A|} \delta^{|B|}$ terms sum to at most one, and all
  the $\hat f(A \cup B)^2$ terms themselves sum to at most one. Therefore, we have bounded the entire sum by
  $O(\sqrt\delta)$ as desired.
\end{proof}

\subsection{Integrality gap for Unique Label Cover SDP}

\paragraph*{Problem and SDP relaxation.} In the Unique Label Cover problem, we are given a label set $L$ and a weighted
multigraph $G = (V, E)$ whose edges are labeled by permutations $\{ \pi_e \colon L \to L \}_{e \in E}$, and are asked to
find an assignment $f \colon V \to L$ of labels to edges that maximizes the fraction of edges $e\{u,v\}$ that are
``consistent'' with our labeling, i.e., $\pi_e(f(u)) = f(v)$. If there exists a labeling that satisfies all the edges,
then it is easy to find such a labeling. However, when all we can guarantee is that $99\%$ fraction of the edges can be
satisfied, it is not known how to find a labeling satisfying even $1\%$ of them. At the same time, present techniques
cannot show that finding a $1\%$-consistent labeling is NP-hard.

One approach to solving this problem is to use an extension of the Goemans-Williamson SDP for Max-Cut, where we set up a
vector $v_i$ for every vertex $v$ and label $i$:
\begin{align*}
  \text{maximize } & \textstyle\E_{e\{u,v\}} \sum_{i \in L} \langle u_i, v_{\pi_e(i)} \rangle \\
  \text{subject to } & \langle u_i, v_j \rangle \geq 0 & \forall u, v \in V, \forall i, j \in L \\
  & \textstyle\sum_{i \in L} \langle v_i, v_i \rangle = 1 & \forall v \in V \\
  & \langle \textstyle\sum_{i \in L} u_i, \textstyle\sum_{j \in L} v_j \rangle = 1 & \forall u, v \in L \\
  & \langle v_i, v_j \rangle = 0 & \forall v \in V, \forall i \neq j \in L
\end{align*}
(The expectation in the objective is over a distribution where $e\{u,v\}$ is picked with probability proportional to its
weight.) The intent is that $\| v_i \|^2$ should be the probability that $v$ receives label $i$, and $\langle u_i, v_j
\rangle$ should be the corresponding joint probability. It is easy to see that this SDP is a relaxation of the original
problem.

\paragraph*{Gap instance.} In an influential paper, Khot and Vishnoi \cite{KV} constructed an integrality gap for this
SDP: for a label set of size $2^k$ and an arbitrary parameter $\eta \in [0,\frac12]$, a graph whose optimal labeling
satisfies $\leq 1/2^{\eta k}$ fraction of the edges, but for which the SDP optimum is at least $1-\eta$. The
hypercontractive inequality plays a central role in the soundness analysis, which we present below.

Let $\tilde V$ be the set of all functions $f \colon \{-1,1\}^k \to \{-1,1\}$ and $L$ be the Fourier basis $\{ \chi_S
\mid S \subseteq [k] \}$; clearly, $|L| = 2^k$. Observe that $\tilde V$ is an Abelian group under pointwise
multiplication, and $L$ is a subgroup. We take the quotient $V = \tilde V/L$ to be the vertex set. Fix an arbitrary
representative for each coset and write $V = \{f_1 L, f_2 L, \dotsc, f_{|V|} L\}$. We shall define a weighted edge
between every pair of these representative functions, then show how to extend this definition to all pairs of functions,
and finally map these edges to edges between cosets.
\begin{itemize}
\item The edge $\tilde e\{f, g\}$ has weight equal to $\Pr_{h,h'}[(f, g) = (h, h')]$, where $h, h' \in V$ are drawn to
  be $\rho$-correlated on every bit with uniform marginals, where $\rho = 1-2\eta$.
\item With every edge $\tilde e\{f_i, f_j\}$ between representative functions, we associate the identity permutation.
\item A non-representative function acts as if its label is assigned according to its coset's representative. Thus, the
  permutation associated with $\tilde e\{f_i \chi_S, f_j \chi_T\}$ is $\chi_U \chi_S \mapsto \chi_U \chi_T$.
\item In the actual graph under consideration, every edge $\tilde e\{f_i \chi_S, f_j \chi_T\}$ appears as an edge
  $e\{f_i L, f_j L\}$ (with the same permutation and weight).
\end{itemize}

\paragraph*{Soundness analysis.}

Given a labeling $R \colon V \to L$ on the cosets, we consider the induced labeling $\tilde R \colon \tilde V \to L$
given by $\tilde R(f_i \chi_S) = R(f_i L) \chi_S$. From our definitions, it is clear that the objective value attained
by $\tilde R$ is precisely $\Pr_{h,h'}[\tilde R(h) = \tilde R(h')]$, where $h, h'$ are chosen as before. Fix any label
$\chi_S$ and consider the indicator function $\phi \colon \tilde V \to \{0, 1\}$ of functions that $\tilde R$ labels
with $\chi_S$. Since exactly one function in each coset gets labeled $\chi_S$, we know that $\E \phi =
1/2^k$. Therefore,
\[
\Pr_{h,h'}[\tilde R(h) = \tilde R(h') = \chi_S] = \E_{h,h'}[\phi(h) \phi(h')]
= \langle h, T_\rho h \rangle = \| T_{\sqrt \rho} h \|^2_2,
\]
which we can upper-bound (using hypercontractivity) by $\| h \|^2_{1+\rho} =1/2^{\frac{2k}{1+\rho}} \leq 1/2^{\eta k}$.

\bibliography{Quals} \bibliographystyle{alpha}

\end{document}